\documentclass[letterpaper, 10 pt, conference]{ieeeconf}

\usepackage{graphicx}
\usepackage{subcaption}
\usepackage[english]{babel}
\usepackage{url}            
\usepackage{booktabs}       
\usepackage{amsfonts}       
\usepackage{nicefrac}       
\usepackage{microtype}      
\usepackage{lipsum}
\usepackage{amssymb}
\usepackage{amsmath}
\usepackage{mathptmx}
\usepackage{hyperref}       
\usepackage{xcolor}
\usepackage{comment}
\usepackage{soul}

\usepackage{algorithm}
\usepackage[noend]{algpseudocode}
\makeatletter
\def\BState{\State\hskip-\ALG@thistlm}
\makeatother

\usepackage[english]{babel}

\DeclareMathOperator*{\argmin}{arg\,min}

\newcommand{\ie}{\emph{i.e.}}

\newtheorem{proposition}{Proposition}

\newtheorem{definition}{Definition}

\usepackage{comment}

\begin{document}
%
\title{Mitigating Moral Hazard in Cyber Insurance Using Risk Preference Design}
%
%
%

\author{Shutian~Liu
        and~Quanyan~Zhu
\thanks{The authors are with the Department of Electrical and Computer Engineering, Tandon School of Engineering,
        New York University, Brooklyn, NY, 11201 USA (e-mails: sl6803@nyu.edu; qz494@nyu.edu).}       
}

\pagestyle{empty}

\maketitle
\thispagestyle{empty}


\begin{abstract}
Cyber insurance is a risk-sharing mechanism that can improve cyber-physical systems (CPS) security and resilience. The risk preference of the insured plays an important role in cyber insurance markets. With the advances in information technologies, it can be reshaped through nudging, marketing, or other types of information campaigns. In this paper, we propose a framework of risk preference design for a class of principal-agent cyber insurance problems.  It creates an additional dimension of freedom for the insurer for designing incentive-compatible and welfare-maximizing cyber insurance contracts. Furthermore, this approach enables a quantitative approach to reduce the moral hazard that arises from information asymmetry between the insured and the insurer. We characterize the conditions under which the optimal contract is monotone in the outcome. This justifies the feasibility of linear contracts in practice. This work establishes a metric to quantify the intensity of moral hazard and create a theoretic underpinning for controlling moral hazard through risk preference design.  We use a linear contract case study to show numerical results and demonstrate its role in strengthening CPS security.
\end{abstract}


%

\section{Introduction}
\label{sec:intro}
%
%
%
%

Recent years have witnessed significant advances in the development and social impacts of Cyber-Physical System (CPS) technologies, including smart grids, autonomous vehicle systems, and implantable medical devices \cite{humayed2017cyber}.
CPS are complex systems integrating various components and technologies (e.g., Information Technologies and Operational Technologies) to enable mission-critical functions.  
As a result, CPS inherits vulnerabilities from many sources.
These vulnerabilities can be exploited by attackers and inflict CPS users significant losses, and they are hard to eliminate. 
There is a crucial need to reduce the risk of CPS users and increase the resiliency in face of attacks.

Cyber insurance is a financial product that transfers cyber risks between the insurer and the insured. 
Complementary to cyber protection/defense mechanisms, it aims to improve the cyber resiliency of a system by mitigating the financial losses as a result of successful attacks.
Cyber insurance involves a contractual agreement that determines the premium and the coverage of the insurance contract. 
The premium is an amount of fee paid by the insured to the insurer to participate in the insurance program.  
The coverage specifies the amount of loss the insurer will cover. With a properly designed contract, the security of CPS users can benefit from  cyber insurance.

One design framework for cyber insurance relies on a class of principal-agent models \cite{grossman1992analysis}, which capture the information asymmetry between the insurer and the insured. 
The asymmetry leads to the moral hazard phenomenon \cite{holmstrom1979moral} in which the insured users tend to behave more recklessly when interacting with the system and result in severe consequences.
Risk aversion plays an important role in moral hazard. 
It is known that risk-averse behaviors of the user can increase the gap between the principal’s profits obtained from the contract with the full observability of the user's action and the contract when the user's action is hidden.  
This phenomenon is absent when the insurer meets a risk-neutral user \cite{chade2002risk}. 
Hence, measuring the level of risk aversion is essential in investigating the effect of moral hazard on cyber insurance.

Risk measure \cite{artzner1999coherent} provides a convenient way to quantify the risks and model agents’ risk perception or preferences. 
Different from the risk described under the expected utility theory, risk measure is richer in capturing individual perceptions of risks induced by the probabilistic nature of random outcomes. 
The design of risk measure is a way to control the risk preferences of the agents using exogenous influences such as nudging, marketing, and propaganda \cite{slovic2006risk, slovic1995construction}.

\begin{figure}[ht]
\centering
\vspace{-4.2mm}\includegraphics[width=0.44\textwidth]{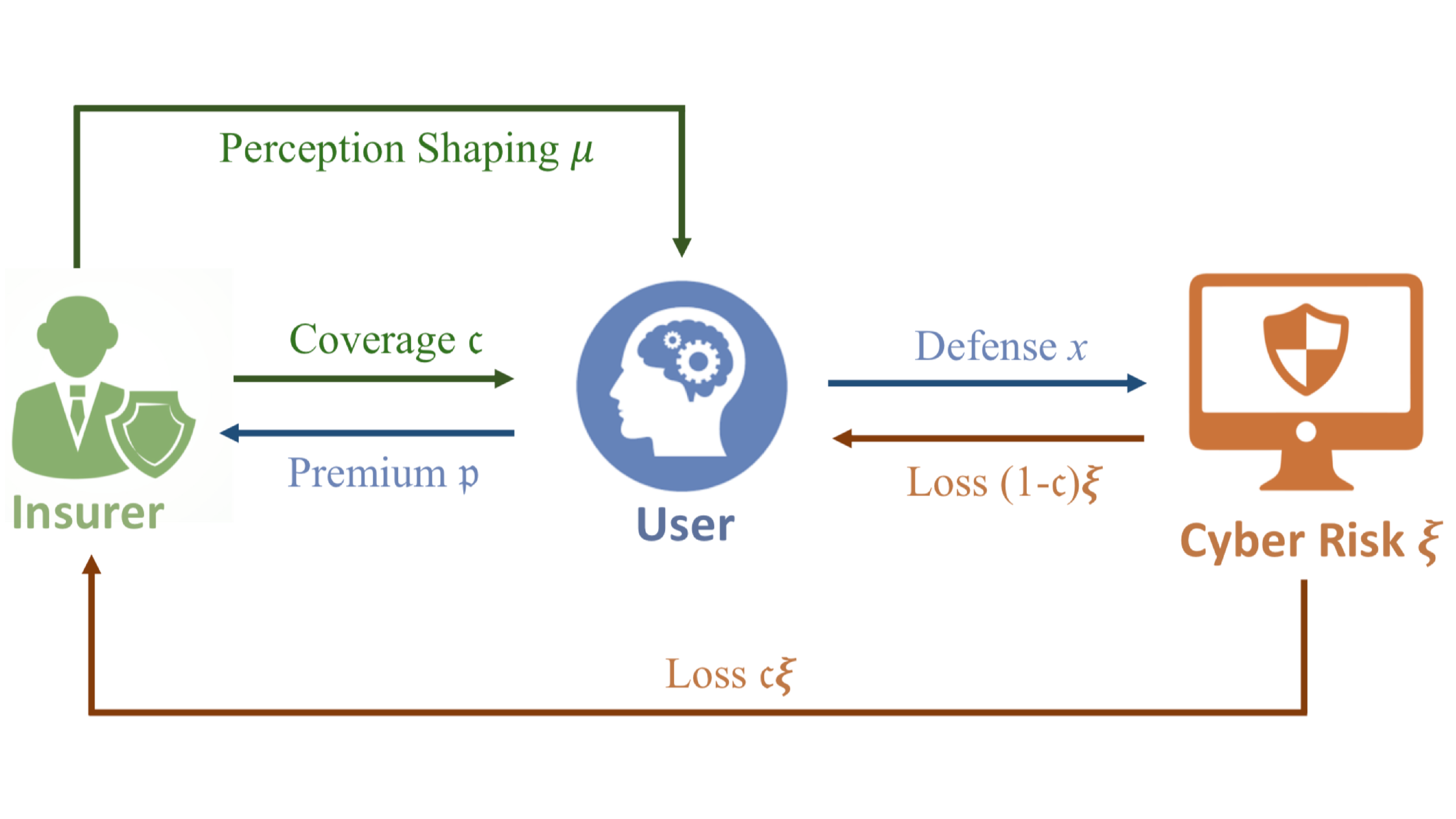}

\vspace{-5mm}\caption[Optional caption for list of figures]{Risk Preference Design Framework for Cyber Insurance:  An insurer aims to reshape the risk perception of the users and design  incentive-compatible cyber insurance contracts. The loss from the cyber risks is covered by the insurer after an active defense from the insured user. } 
\label{fig:framework}
\end{figure}

This work aims to propose, as shown in Fig. \ref{fig:framework}, a risk preference design framework  to shape the users' risk preferences for high-performance cyber insurance.   
We consider a class of principle-agent problems in which one principal meets an idiosyncratic individual from a population of agents. 
Each agent has a risk preference type associated with an idiosyncratic risk measure. 
The goal of the design aims to optimize the profit of the insurer. 
We formulate the problem as a bilevel optimization problem. In the upper-level insurer’s problem, the insurer optimizes her profit by offering a contract and performing risk preference design subject to participation constraints. 
In the lower-level user’s problem, the user follows the claimed contract under the chosen risk preference.

Using the risk preference design framework, we characterize the solution to the insurer’s problem by imposing the conditions under which the optimal contract exhibits monotonicity.
This paves the way for considering linear contracts in practice.
We quantify the Intensity of Moral Hazard (IMH) by benchmarking the action of the user from solving his problem with the action obtained from the full-information counterpart of the insurer’s problem.
Leveraging the optimality conditions of the contract problems, we quantify how IMH varies according to the change in the risk preference type distributions.
This analysis enables the risk preference design for mitigating moral hazard and creates a revenue-compatible design for the principal. 
In the case study, we use linear contracts as an example to show that risk preference design reduces moral hazard and improves security.

Section \ref{sec:related works} provides a brief literature review. 
Section \ref{sec:problem} introduces risk preference types and formulates the design problem.
Section \ref{sec:analysis} analyzes the optimal contract and the proposed approach to quantify moral hazard. 
Section \ref{sec:case} contains a case study.
Section \ref{sec:conclusion} concludes the paper.

\section{Related Works}
\label{sec:related works}
One of the classic approaches to model risk perception relies on 
the expected utility theory of \cite{morgenstern1953theory}, where decisions of the players or the agents are assumed to maximize their utility.
The concept is extended in \cite{arrow1971theory} and \cite{pratt1978risk} to the well-known Arrow-Pratt measure of absolute risk aversion.
Another related approach is the cumulative prospect theory of \cite{tversky1992advances}, where the probability of the random outcomes is weighted.
Motivated by the cumulative prospect theory,  
we use risk measures to capture individual risk preferences.
The seminal work \cite{artzner1999coherent} has proposed to adopt coherent risk measures (CRMs) for risk quantification.
The axioms that a CRM satisfies not only make it meaningful in practice but also enables the dual representation \cite{artzner1999coherent, ruszczynski2006optimization}.
This representation shows the robustness of a CRM and connects it with the popular distributionally robust optimization \cite{shapiro2017distributionally}.

Risk design is closely related to system robustness \cite{liu2020robust}.
Risk design in cyber insurance elicits human behaviors that have an impact on cyber security and resilience.
In \cite{khalili2018designing}, the authors have proposed a pre-screening method 
to increase network security by designing incentive contract which increases users' effort levels.
In \cite{kuru2017effect}, the authors have
have shown that an increased level of risk aversion increases the social welfare gain from using cyber insurance for protection.
This work aims to consolidate the design of risk preferences of users into cyber insurance to further mitigate cyber risks.

\section{Problem Formulation}
\label{sec:problem}
In this section, we introduce the framework of cyber insurance with risk preference design.
We  present the concept of risk preference types to model different perceptions of risks in connection with a class of principal-agent problems.

\subsection{Risk preference types}
\label{sec:problem:risk}
Consider a mass $1$ population of users of networked systems.
This population faces uncertainties that arise from cyber risks captured by the probability space $(\Xi, \mathcal{F})$ with the reference probability measure $P$.
We use $\xi\in\Xi\subset\mathbb{R}$ to denote a sampled outcome.
We denote by $Z:\Xi\rightarrow \mathbb{R}$ a random loss, where 
$Z$ is a measurable function with finite $p$-th order moment from the space $\mathcal{Z} :=\mathcal{L}_p(\Xi, \mathcal{F}, P)$.
The parameter $p$ lives in $[1,+\infty)$.
Individuals in this population possess different risk preferences toward uncertainty.
In particular, each individual is identified with a risk preference type $\theta\in\Theta\subset\mathbb{R}$.
A risk preference of type $\theta$ determines the way user's perception of risks, and it is captured by a risk measure $\rho_\theta:\mathcal{Z}\rightarrow \mathbb{R}$.

Let $(\Theta, \mathcal{G})$ denote the underlying probability space of risk preference types.
The risk preference types in the population are distributed according to a probability measure $\mu^0$ from the set of probability measures $\mathcal{Q}$ defined on $(\Theta, \mathcal{G})$.
Note that, in practice, $\mu^0$ can be known by collecting information through questionnaires or analyzing historical customer behaviors.

The goal of risk preference design is to find a distribution $\mu\in\mathcal{Q}$ so that incentive contracts can be created between the users and the insurer to cover potential losses due to users' cyber risks. 
We consider designing the distribution of risk preferences owing to the following three reasons.
First, the utility of the insurer is measured by the outcome of a population of users participating in the contract. Hence the expected utility captures the incentives of the insurer. Second, it is difficult to design customized contracts as precise individual user's risk perception is hard to measure. 
The collective pattern of the users of the same type can be statistically measured. 
Third, this distributional perspective can be endowed with an interpretation that the insurer and the same user interact repeatedly for a sufficiently long period of time. 
We will elaborate this point further in Section \ref{sec:problem:contract problem}.

We assume that the risk measures $\rho_\theta$ for all $\theta\in\Theta$ are CRMs \cite{artzner1999coherent}. 
Below is the definition of coherent risk measures. 
\begin{definition}[Coherent Risk Measures \cite{artzner1999coherent}]
A function $\rho:\mathcal{Z}\rightarrow\mathbb{R}$ is called a Coherent Risk Measure if it satisfies
\\
(A1) Monotonicity: $\rho(Z)\geq \rho(Z')$ if $Z, Z'\in\mathcal{Z}$ and $Z\succeq Z'$.\\
(A2) Convexity: 
    $\rho(tZ+(1-t)Z') \leq t\rho(Z)+(1-t)\rho(Z')$
for all $Z, Z'\in\mathcal{Z}$ and $t\in[0,1]$.\\
(A3) Translation equivariance: If $Z\in\mathcal{Z}$ and $a\in\mathbb{R}$, then $\rho(Z+a)=\rho(Z)+a$.\\
(A4) Positive homogeneity: $\rho(tZ)=t\rho(Z)$ if $Z\in\mathcal{Z}$ and $t\in\mathbb{R}_+$.
\end{definition}

One of the most significant consequences of CRMs is their dual representation \cite{artzner1999coherent,ruszczynski2006optimization}.
Let $\mathcal{Z}^*:=\mathcal{L}_q(\Xi, \mathcal{F}, P)$ for $q\in (1,\infty]$ denote the dual space of $\mathcal{Z}$, \ie, $\frac{1}{p}+\frac{1}{q}=1$.
Following \cite{artzner1999coherent,ruszczynski2006optimization}, the dual representation of $\rho_\theta$ is
\small
\begin{equation}
    \rho_\theta[Z(\xi)]=\sup_{\zeta \in \mathfrak{A}_\theta}
    \int_{\Xi}Z(\xi)\zeta(\xi)dP(\xi,x),
    \label{eq:dual representation of individual risk}
\end{equation}
\normalsize
where $\mathfrak{A}_\theta \subset \mathcal{Z}^*$ denotes the dual set associated with the risk measure $\rho_\theta$.
The set $\mathfrak{A}_\theta$ is convex and compact when $p\in[1,+\infty)$ \cite{shapiro2021lectures}. 
Hence, the maximum of (\ref{eq:dual representation of individual risk}) is attained.

\subsection{Cyber Insurance with Risk Preference Design}
\label{sec:problem:contract problem}
We incorporate the risk preference types discussed in Section \ref{sec:problem:risk} into a principal-agent problem by specifying the following action-outcome relation.
Consider a principal-agent problem with one principal and one agent.
Let $x\in X$ denote the agent's action, which represents the amount of security investment that he spends on his device.
We set $\Xi:=[\underline{\xi},\Bar{\xi}]$ and identify it as the set of random cyber risks which affect the costs of both the principal and the agent.
The parameterized probability measure $P(\xi, x)$ describes the stochastic relation between the cyber risk and the security investment.
We assume that for all $x\in X$, $P(\xi, x)$ is absolutely continuous with the reference probability measure $P(\xi)$, \ie, there is a density function $p(\xi, x)$ such that $dP(\xi, x)=p(\xi, x)dP(\xi)$.
And we assume that $p(\cdot, x)$ is differentiable for all $x\in X$ and $p(\xi, \cdot)$ is continuously differentiable for all $\xi\in\Xi$.

The insurer acts as the principal in the problem. 
She has the power of claiming a coverage plan of the cyber losses described by the function $w(\cdot):\Xi\rightarrow\Xi$ and choosing a distribution of risk preference types $\mu\in\mathcal{Q}$ for the users.
For a given cyber risk $\xi$, she observes cost (or negative profit) $w(\xi)+\xi$.
We assume that the insurer is risk-neutral.
Hence, her expected cost from claiming a plan $w(\cdot)$ is
$
    \int_{\Xi}w(\xi)+\xi dP(\xi,x).
    \label{eq:cost of plan}
$
The manipulation of risk preferences can come with a cost.
We introduce the order-$1$ Wasserstein distance \cite{villani2009optimal} $W_1(\mu, \mu^0)$
to represent the cost associated with the effort to change the risk preference distribution from $\mu^0$ to $\mu$.
The cost of the insurer in designing $w(\cdot)$ and $\mu$ is
\small
\begin{equation}
    \int_{\Xi}w(\xi)+\xi dP(\xi,x)+\gamma W_1(\mu, \mu^0),
    \label{eq:principal cost}
\end{equation}
\normalsize
where $\gamma>0$ represents the monetary cost of risk manipulation.

Consider a user with type $\theta_i$. 
If he takes action $x_i$ under plan $w(\cdot)$, the cyber risk that he perceives is $\rho_{\theta_i}[U(w(\xi), x_i)]$ with the random cost $U\in\mathcal{Z}$.
We assume that $U$ is increasing and convex in $w$, which indicates that the users exhibit loss aversion.
Note that since $U$ can be understood as a disutility function, our assumptions on $U$ are consistent with the assumptions that a utility function is non-decreasing and concave.
In the literature, these assumptions are often referred to as 'risk aversion' \cite{grossman1992analysis}.
We use the term loss aversion to distinguish it from risk preferences represented by the risk measures $\rho_\theta$.
Our treatment of risk is motivated by \cite{tversky1992advances}, where the authors combine the loss aversion and the probability weighting to characterize decisions under risk.

The agent is assume to be an idiosyncratic individual from the population of users who bears the averaged action $x$.
When the type distribution is $\mu$, the cost of the agent is
\small
\begin{equation}
    \int_{\Theta}\rho_\theta[U(w(\xi),x)]d\mu(\theta).
    \label{eq:agent cost}
\end{equation}
\normalsize
Note that in (\ref{eq:agent cost}), the randomness is addressed by the risk measures $\rho_\theta$; the integral over $\Theta$ represents an averaging over all the risk preference which generates the averaged action $x$.

Let $\Bar{U}>0$ denote the deterministic cyber risk threshold that the agent can bear.
The hidden action cyber insurance problem with risk-preference design is formulated as follows:
\small
\begin{equation}
    \begin{aligned}
         \min_{w(\cdot), \mu(\cdot), x}
         & \int_{\Xi}w(\xi)+\xi dP(\xi,x)+\gamma W_1(\mu, \mu^0) \\
         \text{s.t. } \ \ 
         &\int_{\Theta}\rho_\theta[U(w(\xi),x)]d\mu(\theta)\leq \Bar{U}, \text{(IR)}, \\
         & x\in \argmin_{x'\in X} \int_\Theta \rho_\theta [U(w(\xi),x')]d\mu(\theta), \text{(IC)}. \\
         \label{eq:contract problem}
    \end{aligned}
\end{equation}
\normalsize
In (\ref{eq:contract problem}), (IR) is the individual rationality constraint guaranteeing that participation benefits the agent; (IC) is the incentive compatibility constraint stating that the agent should minimize his own cost when his action is hidden.

The main difference between (\ref{eq:contract problem}) and the classical principal-agent problems is that the distribution $\mu$ characterizing the distribution of risk preferences of the users becomes part of the design, offering a new degree of freedom to model the informational influences on the users and create additional control to elicit proper cyber hygiene of the users and, at the meantime, improve the social welfare and the profit.
Furthermore, as we will show in Section \ref{sec:analysis}, the risk preference design mitigates the moral hazard issue.

\section{Analysis of the Design Framework}
\label{sec:analysis}

\subsection{Characterization of Optimal Contracts}
\label{sec:analysis:characterization}
In this section, we use the first-order approach to characterize several properties of the optimal coverage plan for the contract problems defined in Section \ref{sec:problem}. 
We obtain conditions for the monotonicity of $w(\cdot)$.

Consider the first-order optimality condition of (IC):
\small
\begin{equation}
    \frac{\partial}{\partial x}\int_{\Theta}\rho_\theta[U(w(\xi),x)]d\mu(\theta)= 0.
    \label{eq:foc of IC 1}
\end{equation}
\normalsize
We assume that the set of density functions which attains the maximum of (\ref{eq:dual representation of individual risk}) is a singleton for all $\theta\in\Theta$, \ie, 
$\partial \rho_\theta[U(w(\xi),x)]=\{ \Bar{\zeta}_\theta(\xi, x) \}$.
Then, (\ref{eq:foc of IC 1}) leads to
\small
\begin{equation}
    \begin{aligned}
    0
    & =\frac{\partial}{\partial x}\int_\Theta \int_\Xi U(w(\xi),x)\Bar{\zeta}_\theta(\xi, x)p(\xi, x)dP(\xi)d\mu(\theta) \\
    &=\frac{\partial}{\partial x}\int_{\Xi}\left( \int_{\Theta}\Bar{\zeta}_\theta(\xi, x)d\mu(\theta) \right)U(w(\xi), x)p(\xi,x)dP(\xi)
    .
    \label{eq:foc if IC 2}
    \end{aligned}
\end{equation}
\normalsize
Ignoring the second-order optimality condition of (IC), we find (\ref{eq:foc if IC 2}) as an alternative for the (IC) in (\ref{eq:contract problem}).

The Lagrangian of problem (\ref{eq:contract problem}) when (IC) is substituted by the first-order condition (\ref{eq:foc if IC 2}) can be expressed as the following:
\small
\begin{equation}
    \begin{aligned}
         &\mathbb{L}=  \int_\Xi w(\xi)+\xi dP(\xi, x)+\gamma W_1(\mu, \mu^0) \\
         & \ \ + \alpha\int_\Theta \rho_\theta[U(w(\xi),x)]d\mu(\theta)-\Bar{U} \\
         & \ \ + \beta \frac{\partial}{\partial x} \int_\Theta \mathbb{E}_{\mu}[\Bar{\zeta}_\theta(\xi, x)] U(w(\xi),x)p(\xi, x)dP(\xi)d\mu(\theta). 
         \label{eq:lagrangian of contract}
    \end{aligned}
\end{equation}
\normalsize
Define $\pi(\xi, \theta, x)=\Bar{\zeta}_\theta(\xi, x)p(\xi, x)$.
Minimizing $w(\xi)$ pointwise in $\xi$, we obtain the first-order condition of (\ref{eq:lagrangian of contract}) for $\xi\in\Xi$:
\small
\begin{equation}
\begin{aligned}
    & -p(\xi, x)= \mathbb{E}_\mu \bigg\{ \alpha\pi(\xi,\theta,x)\frac{\partial}{\partial w}U(w(\xi), x) \\
    & \quad +\beta \frac{\partial}{\partial w}\left( \pi(\xi, \theta, x)\frac{\partial}{\partial x}U(w(\xi), x)+U(w(\xi), x)\frac{\partial}{\partial x}\pi(\xi, \theta, x) \right) \bigg\} .
    \label{eq:foc of lagrangian 1}
\end{aligned}
\end{equation}
\normalsize
Before we proceed, we make the assumption that 
$\frac{\partial^2}{\partial x \partial w}U(w(\xi), x)=0$.
Note that this assumption is weaker than the one that the agent's cost is separable in the
security investment and the gain from the coverage plan in the cyber insurance contract.
After several steps of algebraic manipulations, (\ref{eq:foc of lagrangian 1}) leads to
\small
\begin{equation}
    \begin{aligned}
        &\qquad -p(\xi, x) \\
        & = \mathbb{E}_\mu \left\{ \alpha\pi(\xi,\theta,x)\frac{\partial}{\partial w}U(w(\xi), x)+\beta \frac{\partial}{\partial w}U(w(\xi), x)\frac{\partial}{\partial x}\pi(\xi, \theta, x)\right\} \\
        &= \frac{\partial}{\partial w}U(w(\xi), x)\cdot\mathbb{E}_\mu \left\{ \alpha\pi(\xi,\theta,x)+\beta \frac{\partial}{\partial x}\pi(\xi, \theta, x) \right\}.
        \label{eq:foc of lagrangian 2}
    \end{aligned}
\end{equation}
\normalsize
Dividing both sides of (\ref{eq:foc of lagrangian 2}) by $p(\xi, x)\frac{\partial}{\partial w}U(w(\xi), x)$, we obtain
\small
\begin{equation}
\begin{aligned}
    &\qquad \frac{-1}{\frac{\partial}{\partial w}U(w(\xi), x)} \\
    &=\mathbb{E}_\mu
    \left\{ \alpha \Bar{\zeta}_\theta(\xi, x)+\beta \left( \frac{\partial}{\partial x}\Bar{\zeta}_\theta(\xi, x) +\Bar{\zeta}_\theta(\xi, x)\frac{\frac{\partial}{\partial x}p(\xi, x)}{p(\xi, x)} \right)
     \right\}\\
    &= \mathbb{E}_\mu
    \left\{ \Bar{\zeta}_\theta(\xi, x) \right\}\cdot\left(\alpha+\beta \frac{\frac{\partial}{\partial x}p(\xi, x)}{p(\xi, x)} \right) 
    +\beta \mathbb{E}_\mu
    \left\{\frac{\partial}{\partial x} \Bar{\zeta}_\theta(\xi, x) \right\}.
     \label{eq:foc of lagrangian 3}
     \end{aligned}
\end{equation}
\normalsize
Differentiating with respect to $\xi$ over $\Bar{\Xi}$ on both sides of (\ref{eq:foc of lagrangian 3}), we obtain
\small
\begin{equation}
    \begin{aligned}
         & \frac{w'(\xi)\cdot\frac{\partial^2}{\partial w^2}U(w(\xi), x)}{\left(\frac{\partial}{\partial w}U(w(\xi), x)\right)^2}=\beta \frac{\partial}{\partial \xi }\mathbb{E}_\mu
    \left\{ \frac{\partial}{\partial x}\Bar{\zeta}_\theta(\xi, x) \right\} \\
    & \qquad + \frac{\partial}{\partial \xi}
    \left\{ \mathbb{E}_\mu
    [ \Bar{\zeta}_\theta(\xi, x) ]\cdot\left(\alpha+\beta \frac{\frac{\partial}{\partial x}p(\xi, x)}{p(\xi, x)} \right)   \right\}.
     \label{eq:foc of lagrangian 4}
    \end{aligned}
\end{equation}
\normalsize

The above discussions lead to the following result, which characterizes the monotonicity of the optimal coverage plan.

\begin{proposition}
\label{prop:monotonicity}
Suppose that the density function $\Bar{\zeta}_{\theta}(\cdot,x)$ is differentiable on a subset $\Bar{\Xi}\subset\Xi$.
Then, the optimal coverage plan $w$ of (\ref{eq:contract problem}) is increasing in the cyber risk on $\Bar{\Xi}$, \ie, $w'(\xi)\geq 0$, if the following conditions hold:\\
\hspace*{1em}
(C1) The functions $p(\xi, x)$ and $\Bar{\zeta}_\theta(\xi, x)$ satisfy $\frac{\partial}{\partial x}\int_{\Xi}p(\xi, x)d\xi>0$ and $\frac{\partial}{\partial x}\int_{\Xi}\Bar{\zeta}_{\theta}(\xi, x)d\xi>0$, respectively;\\
\hspace*{1em}
(C2) Severe loss avoidance:
\small
\begin{equation*}
     \frac{\partial}{\partial \xi }\mathbb{E}_\mu
    \left\{ \frac{\partial}{\partial x}\Bar{\zeta}_\theta(\xi, x) \right\} \geq 0;
\end{equation*}
\normalsize
\hspace*{1em}(C3) Risk-sensitive monotone likelihood ratio property:
\small
\begin{equation*}
    \frac{\partial}{\partial \xi}
    \left\{ \mathbb{E}_\mu
    [ \Bar{\zeta}_\theta(\xi, x) ]\cdot\left(\alpha+\beta \frac{\frac{\partial}{\partial x}p(\xi, x)}{p(\xi, x)} \right)   \right\} \geq 0.
\end{equation*}
\normalsize
\end{proposition}
\begin{proof}
By combining (C1) with the convexity of $U$ with respect to $w$, we can show $\beta>0$ using the similar arguments as in \cite{holmstrom1979moral}. 
When (C2) and (C3) hold, the optimality condition (\ref{eq:foc of lagrangian 4}) leads to the monotonicity of $w$.
\end{proof}

The conditions in Proposition \ref{prop:monotonicity} have the following interpretations. 
Condition (C1) implies two facts. 
First, without taking the user's risk preference into account,
the true likelihood for higher cyber risks to occur decreases when the security investment increases.
Second, the perceived likelihood for higher cyber risks to occur of a user with risk preference type $\theta$ is decreased when the security investment increases.
These two facts are consistent within the cyber insurance context.
Before we elaborate on condition (C2), we first observe the following fact.
Since the coverage plan $w$ is increasing in the cyber risk $\xi$ and the cost $U$ is increasing in $w$, we conclude that $U$ is also increasing in $\xi$.
In addition, from (\ref{eq:dual representation of individual risk}), we know that for all $\theta$ the optimal density function $\Bar{\zeta}_{\theta}(\xi, x)$ is increasing in $\xi$, \ie, $\frac{\partial}{\partial \xi}\Bar{\zeta}_{\theta}(\xi, x)>0$. 
Then, condition (C2) means that an increase in the user's security investment shows his panic towards the potential occurrences of severer cyber risks. 
In other words, the user tries to avoid severe losses.
Condition (C3) is the risk-sensitive counterpart of the monotone likelihood ratio property:
\small
\begin{equation}
    \frac{\partial}{\partial \xi}\left( \frac{\frac{\partial}{\partial x}p(\xi, x)}{p(\xi, x)} \right)\geq 0.
    \label{eq:monotone likelihood}
\end{equation}
\normalsize

Note that the conditions for proving monotonicity on $\Bar{\Xi}$ can reduce to simply (C1) and the standard monotone likelihood ratio property for some commonly used CRMs, such as the mean semideviation measure and the average value-at-risk measure \cite{ruszczynski2006optimization}.
The reason lies in that the density function $\Bar{\zeta}_\theta(\xi, x)$ of these two measures is piecewise constant. 
It means that the right-hand side of (\ref{eq:foc of lagrangian 4}) becomes 
\small
\begin{equation*}
 \mathbb{E}_\mu
    [ \Bar{\zeta}_\theta(\xi, x) ]\cdot\frac{\partial}{\partial \xi}\left(\alpha+\beta \frac{\frac{\partial}{\partial x}p(\xi, x)}{p(\xi, x)} \right).
\end{equation*}
\normalsize
Then, we can recover (\ref{eq:monotone likelihood}).

\subsection{Mitigation of moral hazard}
\label{sec:analysis:mitigation}
It is well-known from the literature that one of the causes of moral hazard is the issue of hidden action; i.e., the insurer cannot observe the action taken by the insured.
In this section, we propose to measure the IMH and provide a method to mitigate moral hazard.
For notational simplicity, we consider an $n$-dimensional finite risk preference type space $\Theta=\{\theta_1,\theta_2,...,\theta_n\}$ with probability measure $\mu=(\mu_1,\mu_2,...,\mu_n)^T$.

\paragraph{Measuring the IMH}
We first consider the following full-information benchmark  associated with (\ref{eq:contract problem}) assuming that the risk preference distribution $\mu$ is given:
\small
\begin{equation}
    \begin{aligned}
         \min_{w(\cdot), x}
         & \int_{\Xi}w(\xi)+\xi dP(\xi,x)+\gamma W_1(\mu, \mu^0) \\
         \text{s.t. } \ \ 
         &\sum_{i=1}^n \rho_{\theta_i}[U(w(\xi), x)]\cdot \mu_i \leq \Bar{U}, \text{(IR)}. 
         \label{eq:full info benchmark}
    \end{aligned}
\end{equation}
\normalsize
Let $w^*$ and $x^*$ denote the optimal solution of (\ref{eq:full info benchmark}) given $\mu$.
The Lagrangian of (\ref{eq:full info benchmark}) is
\small
\begin{equation}
\begin{aligned}
    \mathbb{L}_b&=\int_{\Xi}(w(\xi)+\xi)p(\xi, x)dP(\xi)+\gamma W_1(\mu, \mu^0) \\
    & \ \ +\alpha\left( \sum_{i=1}^n \rho_{\theta_i}[U(w(\xi), x)]\cdot \mu_i-\Bar{U} \right).
    \label{eq:lagrangian of benchmark}
\end{aligned}
\end{equation}
\normalsize
Then, the action $x^*$ solves
\small
\begin{equation}
\begin{aligned}
    0=\frac{\partial}{\partial x}\mathbb{L}_b
    &=\frac{\partial}{\partial x}\int_{\Xi}(w^*(\xi)+\xi)p(\xi, x)dP(\xi) \\
    & \ \ +\alpha \frac{\partial}{\partial x}\sum_{i=1}^n \rho_{\theta_i}[U(w^*(\xi), x)]\cdot \mu_i.
    \label{eq:FOC bench 1}
\end{aligned}
\end{equation}
\normalsize
Hidden actions refer to the fact that 
$x^*$ will not be the real action taken by the agent after receiving $w^*$ announced by the principal.
From the agent's perspective, an action $x^a$ optimizing his own objective is superior, \ie, $x^a$ such that
\small
\begin{equation}
    x^a \in \argmin_{x\in X} \sum_{i=1}^n \rho_{\theta_i}[U(w^*(\xi), x)]\cdot \mu_i.
    \label{eq:agent problem}
\end{equation}
\normalsize
Therefore, the deviation from $x^*$ to $x^a$ is the result of the moral hazard issue and $x^*-x^a$ is a meaningful metric for measuring the IMH. 
The value of $x^*-x^a$ is nonnegative because of the assumptions made in Section \ref{sec:problem}.
\begin{definition}
\label{def:IMH}
We refer to the quantity $x^*-x^a$ as the IMH.
A high value of $x^*-x^a$ indicates an intense moral hazard.
\end{definition}

Note that (\ref{eq:agent problem}) is part of the IC constraint in (\ref{eq:contract problem}).
However, the actions solving (\ref{eq:contract problem}) and \ref{eq:agent problem} are different in general.

\paragraph{Mitigation of moral hazard through risk preference design}
We study how IMH varies according to the changes in $\mu$.
From (\ref{eq:agent problem}),  $x^a$ satisfies the following first-order condition:
\small
\begin{equation}
    \frac{\partial}{\partial x}\sum_{i=1}^n \rho_{\theta_i}[U(w^*(\xi), x)]\cdot \mu_i=0.
    \label{eq:foc of agent problem}
\end{equation}
\normalsize
Define  $H_1:X\times \mathbb{Q}\rightarrow \mathbb{R}$ and $H_2:X\times \mathbb{Q}\rightarrow \mathbb{R}$ according to
\small
\begin{equation}
    H_1(x, \mu)=\frac{\partial}{\partial x}\mathbb{L}_b,
\end{equation}
\normalsize
as in (\ref{eq:FOC bench 1}), and
\small
\begin{equation}
    H_2(x, \mu)= \frac{\partial}{\partial x}\sum_{i=1}^n \rho_{\theta_i}[U(w(\xi), x)]\cdot \mu_i.
\end{equation}
\normalsize
as in (\ref{eq:foc of agent problem}), respectively.
Then, we obtain from (\ref{eq:FOC bench 1}) and (\ref{eq:foc of agent problem}) that, for a given $\mu$, $H_1(x^*, \mu)=0$ and $H_2(x^a, \mu)=0$.

The next result states how to reduce IMH by designing the distribution of risk preference types.
\begin{proposition}
\label{prop:mitigate IMH}
Let $x^*$ and $x^a$ denote the isolated local minima of (\ref{eq:full info benchmark}) and (\ref{eq:agent problem}) given $\mu$, respectively.
Let $J_{H_1, x}$ denote the Jacobian matrix of $H_1$ with respect to the first argument; let $J_{H_2, x}$ denote the Jacobian matrix of $H_2$ with respect to the first argument.
Then, in a neighborhood $\mathcal{N}_\mu\subset \mathcal{Q}$ of $\mu$, the IMH can be reduced by changing $\mu$ to $\mu'=\mu+c\Delta\mu$ with $c>0$ sufficiently small, if $\Delta\mu$ satisfies
\small
\begin{equation}
    (\Delta\mu)^T\cdot\nabla T(\mu) \leq 0,
    \label{eq:condition of mitigating IMH}
\end{equation}
\normalsize
where $\nabla T(\mu)=(\frac{\partial T}{\partial \mu_1}(\mu), \frac{\partial T}{\partial \mu_2}(\mu),...,\frac{\partial T}{\partial \mu_n}(\mu))^T$ with 
\small
\begin{equation}
\begin{aligned}
    \frac{\partial T}{\partial \mu_i}(\mu)&=\left( J_{H_2,x}(x^a, \mu) \right)^{-1}\frac{\partial H_2}{\partial \mu_i}(x^a, \mu) \\
    & \  \ -
    \left( J_{H_1,x}(x^*, \mu) \right)^{-1}\frac{\partial H_1}{\partial \mu_i}(x^*, \mu)
    ,
\end{aligned}
\end{equation}
\normalsize
for all $i=1,2,...,n$.
\end{proposition}
\begin{proof}
Consider the optimization problems (\ref{eq:full info benchmark}) and (\ref{eq:agent problem}).
Since $x^*$ and $x^a$ are local minima, the first-order conditions hold, \ie, $H_1(x^*, \mu)=0$ and $H_2(x^a, \mu)=0$.
Since the local minima $x^*$ and $x^a$ are isolated, the second-order conditions require that $\frac{\partial^2}{\partial x^2}\mathbb{L}_b\succ 0$ and $\frac{\partial^2}{\partial x^2}\sum_{i=1}^n \rho_{\theta_i}[U(w(\xi), x)]\cdot \mu_i \succ 0$. 
Equivalently, we obtain $J_{H_1, x}(x^*, \mu)\succ0$ and $J_{H_2, x}(x^a, \mu)\succ 0$.
Invoking the implicit function theorem, we know that there exists differentiable functions $h_1(\cdot)$ and $h_2(\cdot)$ defined on neighborhoods $\mathcal{V}_{x^*}\times \mathcal{V}_{\mu}^1$ of $(x^*, \mu)$ and $\mathcal{V}_{x^a}\times \mathcal{V}_{\mu}^2$ of $(x^a, \mu)$, such that $h_1(\mu)=x^*$ and $h_2(\mu)=x^a$.
Furthermore, the derivatives of the functions $h_1(\cdot)$ and $h_2(\cdot)$ with respect to $\mu_i$ for all $ i=1,2,...,n$ are
\small
\begin{equation}
    \frac{\partial h_1}{\partial \mu_i}(\mu)=-\left( J_{H_1, x}(h_1(\mu), \mu) \right)^{-1} \frac{\partial H_1}{\partial \mu_i}(h_1(\mu), \mu),
    \label{eq:derivative of h1}
\end{equation}
\normalsize
and 
\small
\begin{equation}
    \frac{\partial h_2}{\partial \mu_i}(\mu)=-\left( J_{H_2, x}(h_2(\mu), \mu) \right)^{-1} \frac{\partial H_2}{\partial \mu_i}(h_2(\mu), \mu),
    \label{eq:derivative of h2}
\end{equation}
\normalsize
respectively.
Let $\mathcal{N}_\mu=\mathcal{V}_\mu^1\cap\mathcal{V}_\mu^2$.
Define $T(\mu)=h_1(\mu)-h_2(\mu)$ with support $\mathcal{N}_\mu$. 
Then, the quantity $x^*-x^a$ can be  represented using $T(\mu)$ for $\mu\in\mathcal{N}_\mu$.
The gradient $\nabla T(\mu)=( \frac{\partial T}{\partial \mu_i}(\mu),  \frac{\partial T}{\partial \mu_2}(\mu),..., \frac{\partial T}{\partial \mu_n}(\mu))^T$ follows from (\ref{eq:derivative of h1}) and (\ref{eq:derivative of h2}) for $\mu\in\mathcal{N}_\mu$.
Next, consider $T(\mu')$, which represents the IMH when the distribution of risk preference is $\mu'=\mu+c\Delta\mu$.
Taylor's theorem indicates that:
\small
\begin{equation*}
    T(\mu')=T(\mu)+c(\Delta\mu)^T\cdot \nabla T(\mu)+o(c).
\end{equation*}
\normalsize
For sufficiently small $c>0$, condition (\ref{eq:condition of mitigating IMH}) leads to $T(\mu')\leq T(\mu)$.
This completes the proof.
\end{proof}

Obtaining the value of IMH relies on solving problems (\ref{eq:agent problem}) and (\ref{eq:full info benchmark}). 
However, Proposition \ref{prop:mitigate IMH} offers a method to check whether a target distribution of risk preference types mitigates the effect of the moral hazard.

\paragraph{Design methods}
We have shown that the IMH can be mitigated by risk preference design.
Next, we show that the direction $\Delta\mu$ in \ref{eq:condition of mitigating IMH} can decrease the objective value of (\ref{eq:full info benchmark}), making $\mu'=\mu+c\Delta\mu$ a favorable design.

Before stating the results, we first present an equivalent representation of $W_1(\mu, \mu^0)$. 
Since we consider finite probability measures, strong duality holds in the Kantorovich-Rubinstein dual problem \cite{villani2009optimal} of $W_1(\mu, \mu^0)$, \ie,
\small
\begin{equation}
    W_1(\mu, \mu^0)= \max_{b\in B} \sum_{i=1}^n b_i\mu_i-\sum_{i=1}^n b_i\mu_i^0,
    \label{eq:kantorovich dual of W1}
\end{equation}
\normalsize
where $B=\{b\in\mathbb{R}^n: |b_i-b_j|\leq |i-j|, \forall i,j=1,2,...,n \}$.

\begin{proposition}
\label{prop:beneficial design}
For a given $\mu$, let $\nabla T(\mu)$ be defined as in Proposition \ref{prop:mitigate IMH}.
Let $b^*\in\mathbb{R}^n$ denote the solution to (\ref{eq:kantorovich dual of W1}).
Then, there exists $\Delta\mu$, such that by changing $\mu$ to $\mu'=\mu+c\Delta\mu$, the IMH and the objective cost of (\ref{eq:full info benchmark}) decrease, if there does not exist $a>0$, such that 
\small
\begin{equation}
    -ab^*=\nabla T(\mu).
    \label{eq:condition of beneficial design}
\end{equation}
\normalsize
\end{proposition}
\begin{proof}
Since in the objective function of (\ref{eq:full info benchmark}), the only term depending on $\mu$ is $W_1(\mu, \mu^0)$, we only need to show that $W_1(\mu',\mu^0)\leq W_1(\mu, \mu^0)$ for $\mu'=\mu+c\Delta\mu$.
Equivalently, we show that $\Delta\mu$ is a descent direction of $W_1(\mu, \mu^0)$.
To obtain the descent direction, we compute the gradient of $W_1(\mu, \mu^0)$ using the representation (\ref{eq:kantorovich dual of W1}).
Since (\ref{eq:kantorovich dual of W1}) is an optimization problem, the envelope theorem indicates that  
\small
\begin{equation}
    \nabla_{\mu}W_1(\mu, \mu^0)=\nabla_{\mu} \left( \sum_{i=1}^n b_i^*\mu_i-\sum_{i=1}^n b_i^*\mu_i^0 \right)=b^*,
\end{equation}
\normalsize
where $b^*$ solves (\ref{eq:kantorovich dual of W1}).
Then, $\Delta\mu$ descends $W_1(\mu, \mu^0)$ if 
\small
\begin{equation}
    (\Delta\mu)^T\cdot b^*\leq 0.
    \label{eq:condition of decrease objective}
\end{equation}
\normalsize
From Proposition \ref{prop:mitigate IMH}, we know that changing $\mu$ to $\mu'$ mitigates the IMH if $\Delta\mu$ satisfies condition (\ref{eq:condition of mitigating IMH}).
Hence (\ref{eq:condition of mitigating IMH}) and (\ref{eq:condition of decrease objective}) lead to the condition (\ref{eq:condition of beneficial design}).
This completes the proof.
\end{proof}

The results in Proposition \ref{prop:mitigate IMH} and Proposition \ref{prop:beneficial design} lead to strengthened security for the following reasons.
A lower value of the IMH indicates a smaller difference in the action ideal for the insurer and the action the user chooses in reality. 
It means that there is an increased likelihood to sign a contract where the agreed action generates a reasonable amount of cyber risk that the insurer can bear, and it is not overly challenging for the user to execute.
Furthermore, to come up with such a contract, the insurer needs to prepare for the consumption required for shaping the users' perceptions.
Once the cost of risk design and the cost of coverage do not exceed the payment from the user, a contractual agreement which improves cyber security can be reached. We will elaborate on this point in Section \ref{sec:case}.

\section{Case Study: Linear Contracts}
\label{sec:case}
In this section, we use a special case study of linear contracts to show that risk preference design strengthens cyber security. 
We consider a scenario where an attacker infects a network with ransomware \cite{ge2022accountability}.
The network users can defend against the attacker by upgrading their firewalls to decrease the number of files locked by the ransomware.
This reduces the amount of ransom to pay at a cost of a system upgrade.
The users can also buy cyber insurance to cover a portion of their loss by prepaying a premium.

Consider two risk preference types where $\theta_1$ represents the expectation measure, \ie, $\rho_{\theta_1}[Z]=\mathbb{E[Z]}$, and $\theta_2$ represents the absolute semideviation measure \cite{shapiro2021lectures} with parameter $\kappa>0$, \ie, $\rho_{\theta_2}[Z]=\mathbb{E}[Z]+\kappa \mathbb{E}\{[Z-\mathbb{E}[Z]]_+\}$.
The original distribution of types is $\mu^0=(\mu^0_1,\mu^0_2)$.
We consider two available system upgrade investments $X=\{x_L, x_H\}$ with $x_L<x_H$ and three possible losses due to the ransomware $(\xi_1,\xi_2,\xi_3)=(1,2,3)$.
The probability of $(\xi_1,\xi_2,\xi_3)$ is $(0.3,0.4,0.3)$ under $x_L$ and $(0.5,0.3,0.2)$ under $x_H$.
Note that this indicates that when the security investment increases, there is a lower likelihood for more severe losses to occur.

The insurer offers a linear contract containing a coverage $\mathfrak{c}\in(0,1)$ with a premium $\mathfrak{p}>0$.
Her expected cost from this contract is $\mathbb{E}[\mathfrak{c}\xi-\mathfrak{p}]$.
Since $\mathfrak{c}\in(0,1)$, she prefers the high investment $x_H$ to the low investment $x_L$ from the user. We set the user's loss aversion to be quadratic; \ie, he perceives $\xi^2$ when the cyber loss is $\xi$.
The relative cost of investment is $\mathfrak{m}>0$.
If the user does not buy insurance, he receives the full cyber loss $\xi$ without any coverage from the insurer. 
The cost for the user is
$\Bar{U}(x)=\sum_{i=1,2}\mu^0_i\cdot\rho_{\theta_i}[\xi^2+\mathfrak{m}x]$.
If he purchases the insurance, the cost becomes $\Tilde{U}(x)=\sum_{i=1,2}\mu^0_i\cdot\rho_{\theta_i}[((1-\mathfrak{c})\xi)^2+\mathfrak{m}x+\mathfrak{p}]$.

Suppose $\mathfrak{m}(x_H-x_L)>1.1\mathfrak{c}^2+0.6\kappa \mathfrak{c}^2\mu^0_2$.
Then, $\Tilde{U}(x_H)>\Tilde{U}(x_L)$.
From (IC) of (\ref{eq:contract problem}), we know that the user will only agree on $x_L$ in a contract.
The moral hazard issue occurs because of the distinction between the insurer's and user's preferences over the security investments.
Note that participating in the contract is profitable if $\mathfrak{p}\leq (2\mathfrak{c}-\mathfrak{c}^2)\cdot(3.5+1.25\kappa\mu^0_2)$.

With risk preference design, it is straightforward to observe that as long as we increase the portion of type $\theta_2$ to $\mu_2$ such that $\mathfrak{m}(x_H-x_L)<1.1\mathfrak{c}^2+0.6\kappa \mathfrak{c}^2\mu_2$, the user prefers $x_H$ to $x_L$.
It means that by designing the distribution of risk preference types so that more individuals evaluate risk according to the absolute semideviation measure, the moral hazard issue can be mitigated.
One of the features of the absolute semideviation measure is that it emphasizes the high impact outcomes in spite of their low possibilities.
Hence, a user adopting this perception of risk shows aversion toward severe system losses caused by ransomware. 
The consequence of this risk preference design is that a reached contract elicits a high-level security effort from the user to protect the system.
It not only enables risk sharing between the user and the insurer but also enhances the overall defense against  ransomware.

\section{Conclusion}
\label{sec:conclusion}
In this paper, we have proposed cyber insurance with a risk preference design framework to provide risk-sharing contracts for cyber-physical systems security.
We have defined risk preference types for users and adopted risk measures to capture their perceptions of cyber losses. 
We have made the distribution of the types a design parameter in the framework, arming the insurer with an additional dimension of freedom to provide incentive-compatible contracts.
Using the proposed measure for moral hazard, we have shown analytically that risk preference design can not only mitigate the moral hazard issue but also enhance the overall security of the system.

The structure of our framework is essentially a bilevel programming problem. 
These problems are hard to solve due to their structural complexity.
For future work, we would investigate representations of risk measures to simplify the risk preference design problem.
This effort could lead to the equivalents of the agent's problem, enabling single-level reformulations of the bilevel programming problem.

\bibliographystyle{IEEEtran}
\bibliography{bibliography.bib}
\nocite{*}

%






\end{document}